\documentclass{article}
\usepackage[utf8]{inputenc}

\usepackage{rotating}
\usepackage[twoside]{geometry}
\usepackage{epsfig}
\usepackage{indentfirst}
\usepackage[usenames,dvipsnames,svgnames,table]{xcolor}
\usepackage{cmap}

\usepackage{latexsym,amssymb,oldlfont}
\usepackage{amssymb,amsfonts,amstext,amsthm}
\usepackage{mathrsfs}
\usepackage[intlimits]{amsmath}
\newtheorem{definition}{\bf Definition}[section]
\newtheorem{proposition}{\bf Proposition}[section]

\newtheorem{theorem}{\bf Theorem}[section]
\newtheorem{corollary}{\bf Corollary}[section]
\newtheorem{remark}{Remark}[section]

\title{Quantum quasiballistic dynamics and\\ thick point spectrum}
\author{Moacir Aloisio\footnote{UFMG, Avenida Ant\^onio Carlos, 6627, Caixa Postal 702, 30123-970 Belo Horizonte, MG, Brazil.}, Silas L. Carvalh${\rm o}^{*}$\footnote{Corresponding author. Email: silas@mat.ufmg.br Tel: +55 31 3409 5969.}, and C\'esar R. de Oliveira\footnote{UFSCar, Rodovia Washington Lu\'is, km 235, Caixa Postal 676, 13565-905 S\~ao Carlos, SP, Brazil.}}
\date{December 2018}

\begin{document}

\maketitle

\begin{abstract} We obtain dynamical lower bounds for some  self-adjoint operators with pure point spectrum in terms of the spacing properties of their eigenvalues. In particular, it is shown that for systems with thick point spectrum, typically in Baire's sense, the dynamics of each initial condition (with respect to some orthonormal bases of the space) presents a quasiballistic behaviour. We present explicit applications to some Schr\"odinger operators.
\end{abstract} 

\ 

{\flushleft {\bf MSC (2010):} primary 81Q10. Secondary: 28A80, 35J10.}

\  

{\flushleft {\bf Keywords:} point spectrum; spacing of eigenvalues; quasiballistic dynamics.}

\ 

%%%%%%%%%%%%%%%%%%%%%%%%%%%%%%%%%%%%%%%%%%%%%%%%%%%%%%%%%%%%%%%%%%%%%%%%%%%%%%%%%%%%%%%%%%%%%%%%%%%%%%%%%%%%%%%%%%%%%%%%%%%%%%%%%%%%%%%%%%%%%%%%%%%%%%%%%%%%%%%%%%%%%%%%%%%%%%%%%%%%%%%%%%%%%%%%%%%%%%%%%%%%%%%%%%%%%%%%%%%%%%%%%%%%%%%%%%%%%%%%%%%%%%%%%%%--INTRODUCTION--%%%%%%%%%%%%%%%%%%%%%%%%%%%%%%%%%%%%%%%%%%%%%%%%%%%%%%%%%%%%%%%%%%%%%%%%%%%%%%%%%%%%%%%%%%%%%%%%%%%%%%%%%%%%%%%%%%%%%%%%%%%%%%%%%%%%%%%%%%%%%%%%%%%%%%%%%%%%%%%%%%%%%%%%%%%%%%%%%%%%%%%%%%%%%%%%%%%%%%%%%%%%%%%%%%%%%%%%%%%%%%%%%%%%%%%%%%%%%%%%%%%%%%%%%%%%%%%%

\section{Introduction} 

\noindent Let~$T$ be a self-adjoint operator in an infinite dimensional and separable complex Hilbert space~$\mathcal{H}$ and $\xi \in \mathcal{H}$. The relations between the  dynamics $e^{-itT}\xi$ and  spectral properties of~$T$ is a classical subject of the mathematics and physics literatures. Systems with pure point and purely continuous spectra present important qualitative and quantitative dynamical differences. The dense point spectrum in an interval, also called thick point spectrum, is an intermediate step from discrete to continuous spectra, and so one could expect some sort of resemblance of the continuous dynamics to thick point dynamics. However, to the best knowledge of the present authors, no such  property has been yet detailed in the literature. It is the main purpose of this note to show that thick point spectrum implies Baire generically, in the set of initial conditions~$\xi$, a quasiballistic dynamics (in the sense of Definition~\ref{def2}). This applies, in particular, to some arbitrarily small Hilbert-Schmidt perturbations of purely continuous operators whose spectra contain an interval (taking into account the Weyl-von Neumann Theorem \cite{von Neumann,Weyl}). 

It is well known that there are explicit relations between the large time behaviour of the dynamics $e^{-itT}\xi$ and the fractal properties of the spectral measure $\mu_\xi^T$ of~$T$ associated with~$\xi$. In this context, we refer to \cite{Barbaroux1,Guarneri,Last}, among others. In order to obtain the desired quasiballistic dynamics for systems with thick point spectra, we shall explore properties of suitable generalized fractal dimensions of their spectral measures (see Theorem \ref{maintheorem}). 

Let $B=\{\xi_n\}_{n \in {\mathbb{Z}}}$ be an orthonormal basis of~$\mathcal{H}$ and denote, for each $p>0$, the (time-average) $p$-moment of the position operator at time $t>0$, with initial condition~$\xi$, by
\[\langle \langle |X|^p \rangle \rangle_{t,\xi} := \frac{1}{t} \int_0^t \sum_{n \in {\mathbb{Z}}} |n|^p |\langle e^{-isT}\xi ,\xi_n \rangle|^2~{\mathrm d}s.\]
These quantities describe the asymptotic behaviour of the ``basis position'' of the wave packet $e^{-itT}\xi$ as~$t$ goes to infinity (see \cite{Barbaroux1,DamanikDDP,Germinet,Guarneri,Last} and references therein). In order to describe the algebraic growth $\langle \langle |X|^p \rangle \rangle_{t,\xi} \sim t^{\alpha(p)}$ for large t, one usually considers the upper transport exponent, given by
\[ \alpha_B^+(\xi,p) :=  \limsup_{t \to \infty} \frac{\ln \langle \langle |X|^p \rangle \rangle_{t,\xi}}{\ln  t}.\]

The following result, extracted from \cite{Germinet} (see also \cite{DamanikDDP,DamanikDDP2}), gives sufficient conditions for  $\langle\langle |X|^p \rangle \rangle_{t,\xi}$ and  $\alpha_B^+(\xi,p)$ to be well defined.

\begin{proposition}[Propositions 3.1 and 3.2 in \cite{Germinet}] \label{propmoments} Let~$T$ be a bounded self-adjoint operator on $\mathcal{H}$, and let $B=\{\xi_n\}_{n \in {\mathbb{Z}}}$ be an orthonormal basis of~$\mathcal H$. Then, for each $f \in C_c^\infty({\mathbb{R}})$,
\begin{enumerate}
\item [{\rm i)}] $\langle\langle |X|^p \rangle \rangle_{t,f(T)\xi_0}$ is well defined (finite) for all $t,p>0;$ 
\item [{\rm ii)}] $\alpha_B^+(f(T)\xi_0,p)$ are increasing functions of $p;$
\item [{\rm iii)}] $\alpha_B^+(f(T)\xi_0,p) \in [0,p]$, for all $p>0$.
\end{enumerate}
\end{proposition}

\begin{remark}\label{momentsrem}{\rm Proposition \ref{propmoments} is usually stated for ${\mathcal{H}} = \ell^2({\mathbb{Z}})$  and $B= \{\delta_n\}$, the canonical basis of this space. We note that the result in the statement of Proposition~\ref{propmoments} may be reduced to this case. Namely, let $U: {\mathcal{H}} \longrightarrow \ell^2({\mathbb{Z}})$ be a unitary operator given by the law $(U\xi)_n = \langle \xi,\xi_n \rangle$,  and let $\Tilde{T}= UTU^{-1}$ (this is a bounded self-adjoint operator defined on $\ell^2({\mathbb{Z}})$, whose spectral resolution satisfies $P^{\Tilde{T}} = U P^T U^{-1}$). Then, for each $p>0$, $f \in C_c^\infty({\mathbb{R}})$, and $s \in \mathbb{R}$,
\begin{eqnarray}\label{changebase}
 \sum_{n \in {\mathbb{Z}}} |n|^p |\langle e^{-isT}f(T)\xi_0 ,\xi_n \rangle|^2 &=& 
 \sum_{n \in {\mathbb{Z}}} |n|^p |\langle e^{-isUTU^{-1}}f(UTU^{-1})U\xi_0,U\xi_n \rangle|^2 \nonumber\\ &=&   \sum_{n \in {\mathbb{Z}}} |n|^p |\langle e^{-is\Tilde{T}}f(\Tilde{T})\delta_0 ,\delta_n \rangle|^2.
\end{eqnarray}}
\end{remark}

The next result, due to Guarneri and Schulz-Baldes \cite{Guarneri}, establishes a connection between the upper transport exponents $\alpha_B^+(\xi,p)$ and a particular dimensional property of $\mu_\xi^T$.     

\begin{theorem}[Theorem~1 in \cite{Guarneri}]\label{theoGS} Let~$T$ be a self-adjoint operator in $\mathcal{H}$, and let~$B=\{\xi_n\}_{n \in {\mathbb{Z}}}$ be an orthonormal basis. Then, for each $\xi \in \mathcal{H}$ and each $p>0$, 
\begin{equation}\label{pack}
 \alpha_B^+(\xi,p)  \geq   {\rm dim}^+_{{\rm P}}(\mu_\xi^T)p.
\end{equation}
\end{theorem}

In the statement of Theorem \ref{theoGS}, ${\rm dim}^+_{{\rm P}}(\mu_\xi^T)$ represents the (upper) packing dimension of $\mu_\xi^T$ (see Definition~\ref{defHPD} ahead).

The inequality~\eqref{pack} is, in many situations, far from being optimal; del Rio et.\ al.~\cite{Rio} have presented an example of operator in~$\ell^2({\mathbb{Z}})$ with pure point spectrum (and so, with ${\rm dim}^+_{{\rm P}}(\mu_\xi^T)=0$ for every $\xi\in \ell^2({\mathbb{Z}})$) such that  $\alpha_B^+(\delta_0,2)=2$ (here, $B$ is the canonical basis of $\ell^2({\mathbb{Z}})$); see Appendix 2 in~\cite{Rio} for details.

We note that under some additional assumptions on $\mu_\xi^T$, Barbaroux et.\ al.\ \cite{Barbaroux1} have obtained a refinement of this estimate in~(\ref{pack}); namely, they have proven the following result.  

\begin{theorem}[Theorem 2.1 in \cite{Barbaroux1}]\label{theoBGT} Let~$T$ be a self-adjoint operator in $\mathcal{H}$, and let~$B=\{\xi_n\}_{n \in {\mathbb{Z}}}$ be an orthonormal basis. Then, for each $\xi \in \mathcal{H}$ and each $p>0$, 
\begin{equation}\label{gfd}
 \alpha_B^+(\xi,p) \geq   D_{\mu_\xi^{T}}^+\biggr(\frac{1}{1+p}\biggr)p.
\end{equation}
\end{theorem}

In the statement of ~Theorem~\ref{theoBGT}, $D_{\mu_\xi^{T}}^+(q)$ represents the $q$-upper generalized fractal dimension of~$\mu_\xi^T$ (see Definition~\ref{defGFD} ahead). 

\begin{remark} {\rm We stress that inequality~\eqref{gfd} does not depend on the choice of the orthonormal basis of $\mathcal{H}$, and that $\alpha_B^+(\xi,p)=\infty$ may occur.}
\end{remark}
  
This refinement is, in fact, far from trivial, since although every pure point measure has packing dimension equal to zero, some of them may have non-trivial generalized fractal dimensions, as discussed in~\cite{Barbaroux1}. In this sense, it is clear that non-trivial dynamical lower bounds may occur even when the spectrum is pure point. For this reason, during the last decade, many authors have been exploring the relations between dynamical lower bounds and pure point spectrum. We mention the papers \cite{Barbaroux1,Carvalho1,Carvalho2,Oliveira1,Oliveira2,Rio,Howland,SimonBall} for references and additional comments about important results on dynamical lower bounds and pure point spectrum. 

Some words about notation:~$\mathcal{H}$  will always denote an infinite dimensional and separable complex Hilbert space and~$T$ a self-adjoint operator in~$\mathcal{H}$.  For each Borel set $\Lambda \subset \mathbb{R}$, $P^T(\Lambda)$ represents the spectral resolution of~$T$ over~$\Lambda$. A finite Borel measure~$\mu$ on $\mathbb{R}$ is supported on a Borel set $\Lambda \subset \mathbb{R}$ if $\mu({\mathbb{R}} \backslash \Lambda)=0$; we denote the support of~$\mu$ by supp$(\mu)$. In this paper~$\mu$ always indicates  a finite positive Borel measure on $\mathbb{R}$. For each $x \in \mathbb{R}$ and each $\epsilon>0$,  $B(x,\epsilon)$ denotes the open interval $(x-\epsilon,x+\epsilon)$. 

The paper is organized as follows. In Subsection~\ref{subSectMain} we state the main results of this work.  Subsection~\ref{secapplic} is devoted to some of their applications. In Section~\ref{secpre}, we fix some notation and present a dynamical characterization of the generalized fractal dimensions (Proposition~\ref{propGFD1}). In Section~\ref{secmain}, we investigate the existence of suitable dense $G_\delta$ sets of initial conditions (Theorem~\ref{theoGFD}), and then present the proof of Theorem~\ref{maintheorem}. In Appendix~\ref{Appendix}, we present two  simple but important results for this work.

%%%%%%%%%%%%%%%%%%%%%%%%%%%%%%%%%%%%%%%%%%%%%%%%%%%%%%%%%%%%%%%%%%%%%%%%%%%%%%%%%%%%%%%%%%%%%%%%%%%%%%%%%%%%%%%%%%%%%%%%%%5

\subsection{Main results}\label{subSectMain}
\noindent In order to properly present our results, we introduce the following notion.

\begin{definition}\label{def1}\rm Let $(a_j) \subset {\mathbb{R}}$. One says that $(a_j)$ is weakly-spaced if, for each $\alpha>0$, there exists a subsequence $(a_{j_l})$ of  $(a_j)$ such that
\begin{enumerate}
\item[i)] $c_l := a_{j_l}-a_{j_{l+1}}>0$ is monotone and $\displaystyle\lim_{l\to\infty} (a_{j_l}-a_{j_{l+1}}) =0$. 

\item[ii)]  There exists  $C_\alpha > 0$ so that, for every~$l \geq 1$, $a_{j_{l}} - a_{j_{l+1}} \geq C_\alpha/l^{1+\alpha}$.
\end{enumerate}
\end{definition}

\begin{remark}{\rm Let $-\infty < a < b < \infty$. If $\displaystyle\cup_{j} \{a_j\}$ is a dense subset in $[a,b]$, then $(a_j)$ is weakly-spaced (see Proposition~\ref{propWS}).}
\end{remark}

\begin{definition}\label{def2}\rm Let $\xi \in \mathcal{H}$ and let $B=\{\xi_n\}_{n \in {\mathbb{Z}}}$ be an orthonormal basis. The dynamics $e^{-itT}\xi$ is called quasiballistic with respect to $B$ if, for each $p>0$, $\alpha_B^+(\xi,p) = p$. 
\end{definition}

Breuer et.\ al.\ have obtained in \cite{BreuerLS} results on dynamical upper bounds for discrete one dimensional Schr\"odinger operators in terms of various spacing properties of the eigenvalues of their finite volume approximations. In contrast with such results, we show in this work that, under some assumptions on the operator, if the sequence of its eigenvalues is weakly-spaced (Definition~\ref{def1}), then the dynamics of every initial condition in a robust set have quasiballistic behaviour (in the sense of Definition~\ref{def2}). More specifically, we shall prove the following result.

\ 

\begin{theorem}\label{maintheorem} Let $\Lambda \subset \mathbb{R}$ be a bounded Borel set, and let~$T$ be  a self-adjoint operator with pure point spectrum in~$\Lambda$. Suppose that $\Lambda$ contains a weakly-spaced sequence of eigenvalues of~$T$. Then,
\begin{enumerate}
\item [{\rm 1.}]  The set
\[G^T(\Lambda) :=  \{ \xi \in {\mathcal{H}}_\Lambda \mid D_{\mu_\xi^{T}}^+(q) = 1 \;\text{ for each }  0<q<1 \}\]
is a dense $G_\delta$ set in  ${\mathcal{H}}_\Lambda:=P^T(\Lambda)(\mathcal{H})$.
    
\item [{\rm 2.}] For each  $\xi \in G^T(\Lambda)$, each orthonormal basis $B=\{\xi_n\}_{n \in {\mathbb{Z}}}$ so that $\xi_0 := \Vert \xi \Vert^{-1} \xi$, and each $f \in C_c^\infty({\mathbb{R}})$, one has 
\[\alpha_B^+(f(T)\xi_0,p) = p \,\text{ for each }  p>0.\]   
\end{enumerate}
\end{theorem}

\begin{corollary}\label{cormaintheorem}
Let $-\infty<a<b<\infty$, and let~$T$ be a self-adjoint operator with purely thick point spectrum in $I:=[a,b]$. Then,

\begin{enumerate}
\item [{\rm 1.}] The set
\[G^T(I) :=  \{ \xi \in {\mathcal{H}}_I \mid D_{\mu_\xi^{T}}^+(q) = 1 \;\text{ for each }  0<q<1 \}\]
is a dense $G_\delta$ set in ${\mathcal{H}}_I=P^T(I)(\mathcal{H})$.
    
\item [{\rm 2.}] For each $\xi \in G^T(I)$, each orthonormal basis $B=\{\xi_n\}_{n \in {\mathbb{Z}}}$ so that $\xi_0 := \Vert \xi \Vert^{-1} \xi$, and each $f \in C_c^\infty({\mathbb{R}})$, one has 
\[\alpha_B^+(f(T)\xi_0,p) = p \,\text{ for each }  p>0.
\]   

\end{enumerate}
\end{corollary}

\begin{proof} [{Proof} {\rm (Corollary~\ref{cormaintheorem})}] Since the set of eigenvalues of~$T$ contains a dense subset of~$I$, it follows from Proposition~\ref{propWS} that such subset is weakly-spaced. The result is now a direct consequence of Theorem~\ref{maintheorem}.    
\end{proof}

Theorem~\ref{maintheorem} gives a rather general sufficient condition for an operator with pure point spectrum to present non-trivial dynamical lower bounds. The main ingredient in the proof of this result involves a fine analysis of the generalized fractal dimensions of spectral measures of operators with pure point spectrum. Namely, in order to prove Theorem~\ref{maintheorem}, we explore some relations between such dimensions and the spacing properties of its eigenvalues (see Theorem~\ref{theoGFD}).

\begin{remark}
\end{remark}
{\rm 
\begin{enumerate}
 \item[i)] We enumerate the orthonormal basis according to the set of integer numbers. In some specific cases, there are natural enumerations so that our results should be adapted; for instance, in case of $\ell^2({\mathbb Z}^d)$, with $d\ge2$ (see~\cite{Barbaroux1}).

\item[ii)] When~$T$ is a self-adjoint operator with purely thick point spectrum in $I$, it is natural to work with its normalized eigenvectors, say $\{e_n\}$ (namely, $Te_n=\lambda_ne_n$). In this case, for each initial condition $\xi$, each moment \[\langle \langle |X|^p \rangle \rangle_{t,\xi} = \sum_{n \in {\mathbb{Z}}} |n|^p |\langle \xi ,e_n \rangle|^2
\]
is constant over time~$t$, given that, for each $n\in\mathbb{Z}$, $e^{-iTt}e_n=e^{-i\lambda_nt}e_n$. Consequently, if $\sum_{n \in {\mathbb{Z}}} |n|^p |\langle \xi ,e_n \rangle|^2 < \infty$ (this happens when $\xi = e_j$, for instance), it follows from Theorem \ref{theoBGT} that 
\[ D_{\mu_\xi^{T}}^+(q) = 0 \;\text{ for each }  0<q<1\]
and our quasiballistic results do not apply to such vectors. We note that, in this case, $G^T(\Lambda) \subset G(B)$, where 
\[G(B):=\{\xi \in {\mathcal{H}}\mid \langle \langle |X|^p \rangle \rangle_{t,\xi} \equiv \infty   \text{ for all } p>0\}\]
(see also Proposition \ref{propIM}).

We stress that in some sense, we are dealing with exceptional vectors (in the sense that they are ``far from being'' eigenvectors), and our results precise the meaning of such exceptions by selecting bases and initial conditions, always under the hypothesis of the existence of a weakly-spaced sequence of eigenvalues.
\end{enumerate}}

%%%%%%%%%%%%%%%%%%%%%%%%%%%%%%%%%%%%%%%%%%%%%%%%%%%%%%%%%%%%%%%%%%%%%%%%%%%%%%%%%%%%%%%%%%%%%%%%%%%%%%%%%%%%%%%%%%%%%%%%%%%%%%%%%%%%%%%%%%%%%%%%%%%%%%%%%%%%%%%%%%%%%%%%%%%%%%

\subsection{Applications}\label{secapplic}

\noindent Now we illustrate our general results by presenting some explicit applications to Schr\"odinger operators. Namely, we gather below some Schr\"odinger operators with purely thick point spectrum in a bounded interval for which, therefore,  the hypotheses of Corollary~\ref{cormaintheorem} are satisfied.

\

\noindent {\bf Discrete Schr\"odinger operators with uniform electric fields.} Let ${\mathbb{Z}}^d$, $d\geq2$, be endowed with the norm $|k| = \sum_{j=1}^d |k_j|$. The discrete Schr\"odinger operator with uniform electric field of constant strength $E \in {\mathbb{Z}}^d$, $H_v^d$, is defined by the action
\[(H_v^d u)_j := \displaystyle\sum_{|k|=1}u_{j+k} + (E,j)u_j + v_j u_j, \]
where $(\cdot,\cdot)$ denotes the ordinary scalar product in ${\mathbb{R}}^d$ and $\sup_{j}\vert v_j\vert < \infty$. For $d\geq2$, under some assumptions on $(v_j)$ and the direction of the electric field, it is possible to show (see details in~\cite{DinaburgStark}) that~$H_v^d$ has thick point spectrum in~$\mathbb{R}$.  

\

\noindent {\bf Anderson model.} For each fixed $a > 0$, let $\Omega = [-a,a]^{\mathbb{Z}}$ be endowed with the product topology and with the respective Borel $\sigma$-algebra. Assume that $(\omega_j)_{j \in \mathbb{Z}} = \omega \in \Omega$ is a set of independent, identically distributed real-valued random variables with a common probability measure $\rho$ not concentrated on a single point and such that $\int \vert \omega_j \vert^\gamma {\mathrm d}\rho(w_j) < \infty$ for some $\gamma>0$. Denote by $\nu:=\rho^{\mathbb{Z}}$ the probability measure on $\Omega$. The Anderson model is a random Hamiltonian on $\ell^2(\mathbb{Z})$, defined for each $\omega \in \Omega$ by
\[(h_\omega u)_j :=  u_{j-1} + u_{j+1} + \omega_j u_j.\]
It turns out that \cite{DamanikDDP,StolzIntr}  
\[\sigma(h_\omega) = [-2,2] + \text{\rm supp}(\rho),\]
and $\nu$-a.s. $\omega$, $h_\omega$ has thick point spectrum~\cite{carmona,von}. 

\begin{remark}{\rm We note that there are some Anderson operators \cite{StolzIntr} and Anderson Dirac operators \cite{Roberto} defined on $\ell^2({\mathbb{Z}}^d)$, or even in $\mathrm{L}^2({\mathbb{R}}^d)$, with $d\geq 2$, satisfying the hypotheses of Corollary~\ref{cormaintheorem}.}
\end{remark}

\

\noindent {\bf Discrete limit-periodic Schr\"odinger operators.} Let the discrete  Schr\"odinger operator $H_v$, defined on~${\mathbb{Z}}^d$, $d\ge1$, by the action
\[(H_v u)_j := \epsilon \displaystyle\sum_{|k|=1}u_{j+k} + v_j u_j, \]
where $\epsilon$ is a small positive  constant. % and $v = (v_j)$ is a potential.
For some limit periodic potentials $v = (v_j)$, $H_v$ has thick point spectrum in $[0,1]$ (see \cite{Poschel} for details). 

\ 

\noindent {\bf Continuous one-dimensional Schr\"odinger operators.} Consider the continuous one-dimensional Schr\"odinger operator
\[H_V := -\frac{d^2}{dx^2} + V\]
acting in an appropriate domain, where $V$ is a real-valued multiplication operator. 

Let $(k_j)_j$ be an arbitrary sequence of positive real numbers. Then, by Theorem~2 in~\cite{SimonProcams}, there exists a potential~$V$ on $[0,\infty)$ so that $(k_j^2)_j$ are eigenvalues of $H_V$ on $[0,\infty)$. In this case, if $\{k_j^2\}_j$ is dense in $[0,\infty)$, then $H_V$ has thick point spectrum in $[0,\infty)$.   

\ 

\noindent {\bf Quantum Magnetic Hamiltonians.} Let $H(a)$ be the Hamiltonian of a spinless particle moving in two dimensions in a radially symmetric magnetic field $B(r)$, that is,
\[H(a) = (-i\partial/\partial x - a_x)^2 + (-i\partial/\partial y - a_y)^2,\]
where $B = \partial_x a_y - \partial_x a_x$. If there exists  $0<\alpha<1$ such that $B(r) \sim  r^\alpha$ for large $r=\sqrt{x^2+y^2}$, then $H(a)$ has thick point spectrum in $[0,\infty)$. 

%%%%%%%%%%%%%%%%%%%%%%%%%%%%%%%%%%%%%%%%%%%%%%%%%%%%%%%%%%%%%%%%%%%%%%%%%%%%%%%%%%%%%%%%%%%%%%%%%%%%%%%%%%%%%%%%%%%%%%%%%%%%%%%%%%%%%%%%%%%%%%%%%%%%%%%%%%%%%%%%%%%%%%%%%%%%%%%%%%%%%%%%%%%%%%%%%%%%%%%%%%%%%%%%%%%%%%%%%%%%%%%%%%%%%%%%%%%%%%%%%%%--Preliminaries--%%%%%%%%%%%%%%%%%%%%%%%%%%%%%%%%%%%%%%%%%%%%%%%%%%%%%%%%%%%%%%%%%%%%%%%%%%%%%%%%%%%%%%%%%%%%%%%%%%%%%%%%%%%%%%%%%%%%%%%%%%%%%%%%%%%%%%%%%%%%%%%%%%%%%%%%%%%%%%%%%%%%%%%%%%%%%%%%%%%%%%%%%%%%%%%%%%%%%%%%%%%%%%%%%%%%%%%%%%%%%%%%%%%%%%%%%%%%%%%%%%%%%%%%%%%%%%%%%%%%%

\section{Preliminaries}\label{secpre}

\noindent In this section we fix some notation and present auxiliary results.

%%%%%%%%%%%%%%%%%%%%%%%%%%%%%%%%%%%%%%%%%%%%%%%%%%%%%%%%%%%%%%%%%%%%%%%%%%%%%%%%%%%%%%%%%%%%%%%%%%%%%%%%%%%%%%%%%%%%%

\subsection{Fractal dimensions}

\noindent Recall that the pointwise upper scaling exponent of~$\mu$  at $x \in \mathbb{R}$ is defined as  
\[ d_\mu^+(x) := \limsup_{\epsilon \downarrow 0} \frac{\ln \mu (B(x,\epsilon))}{\ln \epsilon},\]
if, for all $\epsilon>0$,  $\mu(B(x;\epsilon))> 0$; if not, $d_\mu^+(x) := \infty$.

\begin{definition}\label{defHPD}{\rm The upper packing dimension of~$\mu$  is defined as}
\[{\rm dim}^+_{{\rm P}} (\mu)  := \mu{\rm {\text -ess}}\sup d_\mu^+(x).\]
\end{definition}

\begin{definition}\label{defGFD}{\rm Let $q\in{\mathbb{R}}\setminus\{ 1\}$. The lower and upper $q$-generalized fractal dimensions of~$\mu$  are defined, respectively, as  
\[D_\mu^-(q) := \liminf_{\epsilon \downarrow 0} \frac{\ln [\int \mu (B(x,\epsilon))^{q-1} {\mathrm d}\mu(x)]}{(q-1) \ln \epsilon} \quad{\rm  and }\quad D_\mu^+(q) := \limsup_{\epsilon \downarrow 0} \frac{\ln [\int \mu (B(x,\epsilon))^{q-1} {\mathrm d}\mu(x)]}{(q-1)\ln \epsilon},\]
with integrals over supp$(\mu)$.}
\end{definition}

\begin{definition}{\rm Let $q\in{\mathbb{R}}\setminus\{ 1\}$. The lower and upper mean-$q$ dimensions of~$\mu$ are defined, respectively, as}  
\[m_\mu^-(q) := \liminf_{\epsilon \downarrow 0} \frac{\ln [\epsilon^{-1} \int_{\mathbb{R}} \mu (B(x,\epsilon))^q \, {\mathrm d}x]}{(q-1) \ln \epsilon} \quad {\rm  and } \quad m_\mu^+(q) := \limsup_{\epsilon \downarrow 0} \frac{\ln [\epsilon^{-1} \int_{\mathbb{R}} \mu (B(x,\epsilon))^q \, {\mathrm d}x]}{(q-1)\ln \epsilon}.\]
\end{definition}

\begin{remark}\label{GFDrem}{\rm If~$\mu$ has  bounded support, then for all $q \in (0,1)$, $0 \leq D_\mu^-(q) \leq D_\mu^+(q) \leq 1$. Moreover, it is possible to show  that for $q>0$, $q \not = 1$, $D_\mu^{\mp}(q) = m_\mu^{\mp}(q)$; see \cite{Barbaroux2} for a detailed discussion.}
\end{remark}

%%%%%%%%%%%%%%%%%%%%%%%%%%%%%%%%%%%%%%%%%%%%%%%%%%%%%%%%%%%%%%%%%%%%%%%%%%%%%%%%%%%%%%%%%%%%%%%%%%%

\subsection{Dynamical characterization of fractal dimensions}

\noindent Let $r>0$ and let~$\mu$ be a finite positive Borel measure on $\mathbb{R}$ so that supp$(\mu) \subset [-r,r]$. Consider, for every $t >0$ and every $q \in \mathbb{R}$,
\[C_\mu(q,t) := t \int_{-r-1}^{r+1} \biggr( \int_{{\mathbb{R}}} e^{-t|x-y|} {\mathrm d}\mu(y) \biggr)^q {\mathrm d}x.\]

\begin{proposition}\label{propGFD1} Let~$\mu$ be as before and $q>0$, $q \not = 1$. Then,
\begin{equation*}
\liminf_{t \to \infty} \frac{\ln C_\mu(q,t)}{(q-1)\ln t} = -D_\mu^+(q),
\end{equation*}
\begin{equation*}
\limsup_{t \to \infty} \frac{\ln C_\mu(q,t)}{(q-1)\ln t} = -D_\mu^-(q).
\end{equation*}
\end{proposition}

Although natural to specialists, we present a proof of this result for the convenience of the reader.

\begin{proof}[{Proof} {\rm (Proposition~\ref{propGFD1})}] We will show that
\begin{equation}\label{GFD1}
\liminf_{t \to \infty} \frac{\ln C_\mu(q,t)}{(q-1)\ln t} = -m_\mu^+(q),
\end{equation}
\begin{equation}\label{GFD2}
\limsup_{t \to \infty} \frac{\ln C_\mu(q,t)}{(q-1)\ln t} = -m_\mu^-(q).
\end{equation}

Since supp$(\mu) \subset [-r,r]$, one has that, for each $t>1$ and each $x \in [-r-1,r+1]^c$,  $\mu(B(x,\frac{1}{t})) = 0$. Hence, it follows that, for  $t>1$,
\begin{eqnarray*} C_\mu(q,t) &=& t \int_{-r-1}^{r+1} \biggr( \int_{{\mathbb{R}}} e^{-t|x-y|} {\mathrm d}\mu(y) \biggr)^q {\mathrm d}x \geq t \int_{-r-1}^{r+1} \biggr( \int_{|x-y|<\frac{1}{t}}  e^{-t|x-y|} {\mathrm d}\mu(y) \biggr)^q {\mathrm d}x \\ &\geq& \frac{t}{e^q} \int_{-r-1}^{r+1} \mu(B(x,\frac{1}{t}))^q {\mathrm d}x = \frac{t}{e^q} \int_{\mathbb{R}} \mu(B(x,\frac{1}{t}))^q {\mathrm d}x
\end{eqnarray*}
and, therefore,
\[\liminf_{t \to \infty} \frac{\ln C_\mu(q,t)}{(q-1)\ln t} \leq -m_\mu^+(q),\qquad\limsup_{t \to \infty} \frac{\ln C_\mu(q,t)}{(q-1)\ln t} \leq  -m_\mu^-(q).\]

Let $0< \delta < 1$. Then, for each $x \in \mathbb{R}$ and  $t>0$,
\begin{eqnarray}
\nonumber \!\!\! \int_{{\mathbb{R}}} e^{-t|x-y|} {\mathrm d}\mu(y) &=&  \int_{|x-y|<\frac{1}{t^{1-\delta}}}  e^{-t|x-y|} {\mathrm d}\mu(y) + \int_{|x-y|\geq\frac{1}{t^{1-\delta}}}  e^{-t|x-y|} {\mathrm d}\mu(y)\\
\nonumber \!\!\!  &\leq&  \mu\big(B\big(x,\frac{1}{t^{1-\delta}}\big)\big) +   e^{-t^\delta} \mu(\mathbb{R}).
\end{eqnarray} 
Therefore,
\begin{eqnarray}\label{DFG3}
 \nonumber \biggr( \int_{{\mathbb{R}}} e^{-t|x-y|} {\mathrm d}\mu(y) \biggr)^q &\leq& 2^q \max\left\{\mu\big(B\big(x,\frac{1}{t^{1-\delta}}\big)\big),\mu({\mathbb{R}}) e^{-t^\delta}\right\}^q\\
  &\leq& 2^q \mu\big(B\big(x,\frac{1}{t^{1-\delta}}\big)\big)^q +   2^q \mu({\mathbb{R}})^q e^{-qt^\delta} .
\end{eqnarray} 
Since $m_\mu^-(q) \geq 0$ (see Remark~\ref{GFDrem}), by (\ref{DFG3}), one gets, for sufficiently large~$t$,
\begin{eqnarray*}
C_\mu(q,t)  &\leq&   2^q t \int_{{\mathbb{R}}} \mu\big(B\big(x,\frac{1}{t^{1-\delta}}\big)\big)^q {\mathrm d}x + (2r + 2) 2^q  \mu({\mathbb{R}})^q t e^{-qt^\delta}\\ &\leq& 2^{q+ 1} t \int_{{\mathbb{R}}} \mu\big(B\big(x,\frac{1}{t^{1-\delta}}\big)\big)^q \,{\mathrm d}x.
\end{eqnarray*}
Thus,
\[(1-\delta)\liminf_{t \to \infty} \frac{\ln C_\mu(q,t)}{(q-1)\ln t} \geq -m_\mu^+(q),\]
\[(1-\delta)\limsup_{t \to \infty} \frac{\ln C_\mu(q,t)}{(q-1)\ln t} \geq  -m_\mu^-(q).\]
Since $0 < \delta < 1$ is arbitrary, the complementary inequalities in~\eqref{GFD1} and~\eqref{GFD2} follow. The results are now a consequence of Remark~\ref{GFDrem}.
\end{proof}

%%%%%%%%%%%%%%%%%%%%%%%%%%%%%%%%%%%%%%%%%%%%%%%%%%%%%%%%%%%%%%%%%%%%%%%%%%%%%%%%%%%%%%%%%%%%%%%%%%%%%%%%%%%%%%%%%%%%%%%%%%%%%%%%%%%%%%%%%%%%%%%%%%%%%%%%%%%%%%%%%%%%%%%%%%%%%%%%%%%%%%%%%%%%%%%%%%%%%%%%%%%%%%%%%%%%%%%%%%%%%%%%%%%%%%%%%%%%%%%%%%%--Lower--bounds--and--fractal--dimensions--%%%%%%%%%%%%%%%%%%%%%%%%%%%%%%%%%%%%%%%%%%%%%%%%%%%%%%%%%%%%%%%%%%%%%%%%%%%%%%%%%%%%%%%%%%%%%%%%%%%%%%%%%%%%%%%%%%%%%%%%%%%%%%%%%%%%%%%%%%%%%%%%%%%%%%%%%%%%%%%%%%%%%%%%%%%%%%%%%%%%%%%%%%%%%%%%%%%%%%%%%%%%%%%%%%%%%%%%%%%%%%%%%%%%%%%%%%%%

\section{Lower bounds and fractal dimensions}\label{secmain}

\noindent In this section, our main goal is to prove Theorem~\ref{maintheorem}. We begin investigating the existence of $G_\delta$ sets.

%%%%%%%%%%%%%%%%%%%%%%%%%%%%%%%%%%%%%%%%%%%%%%%%%%%%%%%%%%%%%%%%%%%%%%%%%%%%%%%%%%%%%%%%%%%%%%%%%%%%%%%%%%%%%%%%%%%%%%%%%%%%%%%%%%%%%%%%%%%%%%%%%

\subsection{$G_\delta$ sets} 

\begin{proposition}\label{propgdelta} Let~$T$ be a bounded self-adjoint operator on~$\mathcal{H}$ and  $q \in (0,1)$. Then, for each $\gamma \geq 0$,
\begin{enumerate}
\item[i)] $G_{\gamma^-}^T := \{\xi \in {\mathcal{H}} \mid  D_{\mu_\xi^T}^-(q) \leq  \gamma \}$ is a $G_\delta$ set in $\mathcal{H}$,
\item[ii)] $G_{\gamma^+}^T := \{\xi \in {\mathcal{H}} \mid D_{\mu_\xi^T}^+(q) \geq  \gamma \}$ is a $G_\delta$ set in $\mathcal{H}$. 
\end{enumerate}
\end{proposition}

\begin{proof} We just present the proof of item~{\it i)}. For each $j \geq 1 $, let $g_j:(0,\infty)\rightarrow(0,\infty)$, $g_j(t) := t^{\frac{1}{j}+\gamma}$. Since, for each $j\geq 1$ and each $t > 0$, the mapping
\[{\mathcal{H}} \ni \xi \mapsto g_j(t) C_{\mu_\xi^T}(q,t)^{1/(q-1)}\]
is continuous (by dominated convergence), it follows that, for each $j\geq 1$,  $k\geq 1$, and  $n\geq 1$, the set
\[\bigcup_{t\geq k} \{\xi \in {\mathcal{H}} \mid g_j(t) C_{\mu_\xi^T}(q,t)^{1/(q-1)} > n \}\]
is open; thus, by Proposition~\ref{propGFD1},

\begin{eqnarray*}
G_{\gamma^-}^T &=& \bigcap_{j\geq 1}  \{\xi \in {\mathcal{H}} \mid \limsup_{t \to \infty} g_j(t) C_{\mu_\xi^T}(q,t)^{1/(q-1)} = \infty \}\\ &=&  \bigcap_{j\geq 1}  \bigcap_{n\geq 1} \bigcap_{k \geq 1} \bigcup_{t\geq k} \{\xi \in {\mathcal{H}} \mid g_j(t) C_{\mu_\xi^T}(q,t)^{1/(q-1)} > n\}
\end{eqnarray*}
is a $G_\delta$ set in $\mathcal{H}$.
\end{proof}
 
%%%%%%%%%%%%%%%%%%%%%%%%%%%%%%%%%%%%%%%%%%%%%%%%%%%%%%%%%%%%%%%%%%%%%%%%%%%%%%%%%%%%%%%%%%%%%%%%%%%%%%%%%%%%%%%%%%%%%%%%%%%%%%%%%%%%%%%%%%%%%%%%%%%%%%%%%%

\subsection{Generic minimal $D_{\mu_\xi^T}^-(q)$ and maximal $D_{\mu_\xi^T}^+(q)$}

\noindent Next, we  relate some spacing properties of the eigenvalues of self-adjoint operators with pure point spectrum to the generalized fractal dimensions of their spectral measures. The typical value of such dimensions (in Baire's sense) is obtained if the sequences of eigenvalues of these operators are weakly-spaced (see Definition~\ref{def1}).

\begin{theorem}\label{theoGFD}Let~$T$ be a bounded self-adjoint operator  on~$\mathcal{H}$ with pure point spectrum. Suppose that the sequence of eigenvalues of~$T$ is weakly-spaced. Then, for each $q \in (0,1)$,  
\[\{ \xi \in {\mathcal{H}} \mid D_{\mu_\xi^T}^-(q) = 0 \;\; and \;\; D_{\mu_\xi^T}^+(q) = 1\}\]
is a dense $G_\delta$ set in $\mathcal{H}$.
\end{theorem}
\begin{proof} Fix $0<q<1$ and let  $(e_j)$ be an orthonormal family of eigenvectors of~$T$, that is, $Te_j = \lambda_j e_j$ for every $j\ge 1$.  

Let $(b_j) \subset \mathbb{C}$ be a sequence such that $|b_j| >0$, for all $j \geq 1$,  and $\sum_{j=1}^\infty \vert b_j \vert^{2q} < \infty$.  Given $\xi \in \mathcal{H}$, write $\xi = \sum_{j=1}^\infty a_j e_j$, and then consider, for each $k \geq 1$, 
\[\xi_k := \displaystyle\sum_{j=1}^k a_j e_j + \displaystyle\sum_{j=k+1}^\infty b_j e_j.\]
It is clear that $\xi_k \rightarrow \xi$. Moreover, for  $k \geq 1$ and each $\epsilon>0$,
\begin{eqnarray}\label{theoGFD1}
\nonumber \int_{{\rm supp}(\mu_{\xi_k}^T)} \mu_{\xi_k}^T(B(x,\epsilon))^{q-1} {\mathrm d}\mu_{\xi_k}^T(x) &=& \displaystyle\sum_{j=1}^\infty \mu_{\xi_k}^T(B(\lambda_j,\epsilon))^{q-1} \mu_{\xi_k}^T (\{\lambda_j\})\\ &\leq& \displaystyle\sum_{j=1}^\infty \mu_{\xi_k}^T (\{\lambda_j\})^q = \displaystyle\sum_{j=1}^k \vert a_j \vert^{2q} + \displaystyle\sum_{j=k+1}^\infty \vert b_j \vert^{2q},
\end{eqnarray}
from which it follows that $ D_{\mu_{\xi_k}^T}^\mp(q) = 0$. Hence,  $G_{0^-}^T = \{ \xi \in {\mathcal{H}} \mid D_{\mu_\xi^T}^-(q) = 0 \}$ is a dense set and, therefore, by Proposition~\ref{propgdelta}, a dense $G_\delta$ set in $\mathcal{H}$.

Now we pass to the upper dimensions. Fix an $n\in\mathbb{N}$ with $n > \frac{q}{1-q}$ and let $(\lambda_{j_l})$ be a subsequence of $(\lambda_j)$ so that: i)~$\lim_{l\to\infty} (\lambda_{j_{l}} - \lambda_{j_{l+1}}) =0$ monotonically;  ii)~there exists a $C_n > 0$ such that, for every~$l \geq 1$, $ \lambda_{j_{l}} - \lambda_{j_{l+1}} \geq C_n/l^{1+\frac{1}{n}}$. Consider, for each $k \geq 1$,
\[\xi_k := \displaystyle\sum_{l=1}^k a_l e_l + \displaystyle\sum_{l=r(k)}^\infty \frac{1}{\sqrt{l^{1+\frac{1}{n}}}} e_{j_l},\]
where we set $r(k)$ large enough so that $\{e_1, ...., e_k, e_{j_{r(k)}}, e_{j_{r(k)+1}},...\}$ is an orthonormal set. Again, $\xi_k \rightarrow \xi$ in $\mathcal{H}$. 

For each $m \geq 1$, put $\epsilon_m := \vert \lambda_{j_{m}} - \lambda_{j_{m+1}} \vert/2$. Then, for each $m>M(k)$ and each $1 \leq l \leq  m$,
\[\mu_{\xi_k}^T(B(\lambda_{j_l},\epsilon_m)) = \mu_{\xi_k}^T (\{\lambda_{j_l}\}),\]
where $M(k)$ is large enough so that for each $m>M(k)$, each $l \geq 1$ and each $1\leq i\leq k$, $\lambda_i \not \in B(\lambda_{j_l},\epsilon_m)$. Hence, for $m > \max \{M(k),r(k)\}=:s(k)$, 

\begin{eqnarray*}
\int_{{\rm supp}(\mu_{\xi_k}^T)} \mu_{\xi_k}^T(B(x,\epsilon_m))^{q-1} {\mathrm d}\mu_{\xi_k}^T(x) &=& \displaystyle\sum_{l=1}^\infty \mu_{\xi_k}^T(B(\lambda_l,\epsilon_m))^{q-1} \mu_{\xi_k}^T (\{\lambda_l\})\\ &\geq& \displaystyle\sum_{l=s(k)}^{m}  \mu_{\xi_k}^T(B(\lambda_{j_l},\epsilon_m))^{q-1} \mu_{\xi_k}^T (\{\lambda_{j_l}\})\\ &=& \displaystyle\sum_{l=s(k)}^{m}  \mu_{\xi_k}^T(\{\lambda_{j_l}\})^q = \displaystyle\sum_{l=s(k)}^{m} \frac{1}{l^{(1+\frac{1}{n})q}}\\ &\geq& E_k \, m^{1-(1+\frac{1}{n})q} \geq E_k\left(\frac{C_n}{2\epsilon_m}\right)^{(1-(1+\frac{1}{n})q)/(1+\frac{1}{n})},
\end{eqnarray*}
where $E_k$ is a constant depending only of $k$, which results in
\[ D_{\mu_{\xi_k}^T}^+(q) \geq  \frac{1-(1+\frac{1}{n})q}{(1-q)(1+\frac{1}{n})} =: t_{n,q}. \]
Thus, $G_{(t_{n,q})^+}^T$  is a dense set and, therefore, by Proposition~\ref{propgdelta}, a dense $G_\delta$ set in $\mathcal{H}$. Since 
\[ G_{1^+}^T = \bigcap_{n > \frac{q}{1-q}} G_{(t_{n,q})^+}^T\]
 and $G_{1^+}^T = \{\xi \in {\mathcal{H}} \mid D_{\mu_\xi^T}^+(q) = 1 \}$ (see Remark~\ref{GFDrem}), the result is proven.
\end{proof}

%%%%%%%%%%%%%%%%%%%%%%%%%%%%%%%%%%%%%%%%%%%%%%%%%%%%%%%%%%%%%%%%%%%%%%%%%%%%%%%%%%%%%%%%%%%%%%%%%%%%%%%%%%%%%%%%%%%%%%%%%%%%%%%%%%%%%%%%%%%%%%%%%%%

\subsection{Proof of Theorem~\ref{maintheorem}}

\noindent Theorem~\ref{maintheorem} is an application of the next result (Theorem~\ref{maintheorem2}) to $P^T(\Lambda)\mathcal H$.

\begin{theorem}\label{maintheorem2} Let~$T$ be a bounded self-adjoint operator with pure point spectrum on~$\mathcal H$. Suppose that  the sequence of eigenvalues of~$T$ is weakly-spaced. Then,
\begin{enumerate}
\item [{\rm 1.}]  The set
\[G^T :=  \{ \xi \in {\mathcal{H}} \mid D_{\mu_\xi^{T}}^+(q) = 1 \;\text{ for each }  0<q<1 \}\]
is a dense $G_\delta$ subset of~${\mathcal{H}}$.
    
\item [{\rm 2.}] For  fixed  $\xi \in G^T$, each orthonormal basis $B=\{\xi_n\}_{n \in {\mathbb{Z}}}$ so that $\xi_0 := \Vert \xi \Vert^{-1} \xi$, and each $f \in C_c^\infty({\mathbb{R}})$, one has 
\[\alpha_B^+(f(T)\xi_0,p) = p, \;\;\text{ for each }\;  p>0.\]   
\end{enumerate}
\end{theorem}

\begin{proof} 

\begin{enumerate}
\item[{\rm 1.}] If follows from Proposition \ref{propgdelta} and Theorem \ref{theoGFD}.
    
\item[{\rm 2.}]  Fix $\xi \in G^T$ and $f \in C_c^\infty({\mathbb{R}})$; it follows from Proposition \ref{propmoments} that 
\[\alpha_B^+(f(T)\xi_0,p) \leq p \,\text{ for each }  p>0.\] 

Since, for each  $\epsilon >0$ and each $x \in {\mathbb{R}}$, one has  
\[\mu_{f(T)\xi_0}^T(B(x,\epsilon)) \leq \displaystyle\sup_{y \in {\mathbb{R}}} |f(y)|^2 \, \mu_{\xi_0}^T(B(x,\epsilon)),\]
it follows that  
\[D_{\mu_{f(T)\xi_0}^{T}}^+(q) \geq  D_{\mu_{\xi_0}^{T}}^+(q) =1, \;\; \forall \,   0<q<1,\]
and therefore  
\[D_{\mu_{f(T)\xi_0}^{T}}^+(q)=1, \;\; \forall \, 0<q<1
\]
(see Remark~\ref{GFDrem}). Consequently,  by Theorem~\ref{theoBGT},
\[\alpha_B^+(f(T)\xi_0,p) \geq  p \,\text{ for each }  p>0.\]
\end{enumerate}
\end{proof}

%%%%%%%%%%%%%%%%%%%%%%%%%%%%%%%%%%%%%%%%%%%%%%%%%%%%%%%%%%%%%%%%%%%%%%%%%%%%%%%%%%%%%%%%%%%%%%%%%%%%%%%%%%%%%%%%%%%%%%%%%%%%%%%%%%%%%%%%%%%%%%%%%%%%%%%%%%%%%%%%%%%%%%%%%%%%%%%%%%%%%%%%%%%%%%%%%%%%%%%%%%%%%%%%%%%%%%%%%%%%%%%%%%%%%%%%%%%%%%%%%%%%%%%%%%%%%--Appendix--%%%%%%%%%%%%%%%%%%%%%%%%%%%%%%%%%%%%%%%%%%%%%%%%%%%%%%%%%%%%%%%%%%%%%%%%%%%%%%%%%%%%%%%%%%%%%%%%%%%%%%%%%%%%%%%%%%%%%%%%%%%%%%%%%%%%%%%%%%%%%%%%%%%%%%%%%%%%%%%%%%%%%%%%%%%%%%%%%%%%%%%%%%%%%%%%%%%%%%%%%%%%%%%%%%%%%%%%%%%%%%%%%%%%%%%%%%%%%%%%%%

\appendix

\section{Appendix}\label{Appendix}

\begin{proposition}\label{propWS} Let $-\infty < a<b <\infty$. If $\displaystyle\cup_{j} \{a_j\}$ is a dense subset of $[a,b]$, then $(a_j)$ is weakly-spaced.
\end{proposition}

\begin{proof} Let $\alpha>0$. Firstly, we note that, for each $x > 1$,
\begin{equation}\label{WS1}
\biggr(\frac{x}{x-1}\biggr)^\alpha + \biggr(\frac{x}{x+1}\biggr)^\alpha > 2.   
\end{equation}
Namely, set
\[f(\alpha) :=  \biggr(\frac{x}{x-1}\biggr)^\alpha + \biggr(\frac{x}{x+1}\biggr)^\alpha.\]
So, 
\begin{eqnarray*}
 \biggr(\frac{x-1}{x}\biggr)^\alpha f'(\alpha) &=&  \ln \biggr(\frac{x}{x-1}\biggr)  - \biggr(\frac{x-1}{x+1}\biggr)^\alpha \ln \biggr(\frac{x+1}{x}\biggr)\\ &>& \ln \biggr(\frac{x}{x-1}\biggr) \biggr(1-\biggr(\frac{x-1}{x+1}\biggr)^\alpha\biggr)>0.
\end{eqnarray*}
Since $f(0) = 2$, the inequality in (\ref{WS1}) follows.   

For each $l \geq 1$, set
\[b_l := a+ \frac{1}{l^\alpha};\]
 by~\eqref{WS1}, for~$l\geq2$ one has $K_l:=b_{l-1} -2b_{l} + b_{l+ 1}>0$.
Note that 
\begin{eqnarray}\label{WS2}
\lim_{l \to \infty} l^{1+\alpha}(b_l -b_{l+1})=\alpha.
\end{eqnarray}
 Now, for~$l$ sufficiently large such that $b_l \in [a,b)$, pick $a_{j_l}$ satisfying 
\begin{equation}\label{WS3}
 0\leq a_{j_l} - b_l \leq \min \biggr\{\frac{K_l}{2},\frac{\alpha}{4l^{1+\alpha}}\biggr\}.   
\end{equation}
Then, by (\ref{WS2}) and (\ref{WS3}), for $l$ sufficiently large, one has 
\begin{eqnarray*}
a_{j_l} - a_{j_{l+1}} &=& (a_{j_l} - b_l) - (a_{j_{l+1}}-b_{l+1}) + (b_l - b_{l+1})\\ &\geq& -\frac{\alpha}{4(l+1)^{1+\alpha}} + \frac{3\alpha}{4l^{1+\alpha}} \geq \frac{\alpha}{2l^{1+\alpha}},
\end{eqnarray*}
\begin{eqnarray*}
a_{j_l} - a_{j_{l+1}} &=& (a_{j_l} - b_l) - (a_{j_{l+1}}-b_{l+1}) + (b_l - b_{l+1})\\ &\leq& \frac{\alpha}{4l^{1+\alpha}} + \frac{7\alpha}{4l^{1+\alpha}} = \frac{2\alpha}{l^{1+\alpha}}.
\end{eqnarray*}
Hence,
\[\frac{\alpha}{2l^{1+\alpha}} \leq a_{j_l} - a_{j_{l+1}} \leq  \frac{2\alpha}{l^{1+\alpha}}.\] 

Moreover,
\begin{eqnarray*}
(a_{j_l} -a_{j_{l+1}})-(a_{j_{l+1}}- a_{j_{l+2}}) &=& (a_{j_l} -2a_{j_{l+1}} + a_{j_{l+2}}) \\ &=& a_{j_l} -b_l -2(a_{j_{l+1}} - b_{l+1}) + a_{j_{l+2}} -b_{l+2} + (b_{l} -2b_{l+1} + b_{l+2})\\ &\geq& -2(a_{j_{l+1}} - b_{l+1}) + K_{l+1} \geq 0, 
\end{eqnarray*}
which implies that $a_{j_l} - a_{j_{l+1}}$ goes to zero monotonically. Therefore, $(a_j)$ is weakly-spaced.
\end{proof}

\begin{proposition}\label{propIM} Let $T$ be a self-adjoint operator in $\mathcal{H}$,  $p>0$ and let $B = \{e_n\}_{n \in {\mathbb{Z}}}$ be an orthonormal basis. Suppose that, for every $\xi \in \mathcal{H}$, $\sum_{n \in {\mathbb{Z}}} |n|^p |\langle e^{-itT}\xi ,e_n \rangle|^2 = \infty$ for  $t = 0$ if, and only if,  $\sum_{n \in {\mathbb{Z}}} |n|^p |\langle e^{-itT}\xi ,e_n \rangle|^2 = \infty$ for all $t \in {\mathbb{R}}$. Then, 
\[G(B) = \{\xi \in {\mathcal{H}}\mid \langle \langle |X|^p \rangle \rangle_{t,\xi} \equiv \infty   \text{ for all } p>0\}\]
is a dense $G_\delta$ set in ${\mathcal{H}}$.
\end{proposition}

\begin{proof} One just has to show that    
\[\{\xi \in {\mathcal{H}}\mid \sum_{n \in {\mathbb{Z}}} |n|^p |\langle \xi ,e_n \rangle|^2 = \infty   \text{ for all } p>0\}\] is a dense $G_\delta$ set in ${\mathcal{H}}$.

Since for each $j\geq 1$, the mapping 
\[{\mathcal{H}} \ni \xi \longmapsto \sum_{|n|\leq j} |n|^p |\langle \xi ,e_n \rangle|^2\]
is continuous, it follows that, for each $p>0$,
\begin{eqnarray*}
\{\xi \in {\mathcal{H}} \mid \sum_{n \in {\mathbb{Z}}} |n|^p |\langle \xi ,e_n \rangle|^2 = \infty\} =    \bigcap_{k \geq 1} \bigcap_{j\geq 1} \{ \xi \in {\mathcal{H}} \mid \sum_{|n|\leq j} |n|^p |\langle \xi ,e_n \rangle|^2 > k\}
\end{eqnarray*}
is a $G_\delta$ set in $\mathcal{H}$. Now, for each fixed $p>0$, $\xi \in \mathcal{H}$ and $j\in\mathbb{N}$, set 
\[\xi_j := \displaystyle\sum_{|n|\leq j} \langle \xi ,e_n \rangle e_n + \displaystyle\sum_{|n|>j}^\infty  \frac{1}{\sqrt{|n|^{p+1}}}e_n.\]
It is clear that $\xi_j \rightarrow \xi$ in $\mathcal{H}$. Moreover, for each $j \geq 1$, $\sum_{n \in {\mathbb{Z}}} n^p |\langle \xi_j ,e_n \rangle|^2 = \infty$. Thus, for each $p>0$, $\{\xi \in {\mathcal{H}} \mid \sum_{n \in {\mathbb{Z}}} |n|^p |\langle \xi ,e_n \rangle|^2 = \infty\}$ is dense in ${\mathcal{H}}$, and therefore, by Baire's Theorem,  
\[\{\xi \in {\mathcal{H}}\mid \sum_{n \in {\mathbb{Z}}} |n|^p |\langle \xi ,e_n \rangle|^2 = \infty   \text{ for all } p>0\} = \bigcap_{p\in {\mathbb{Q}}^+} \{\xi \in {\mathcal{H}} \mid \sum_{n \in {\mathbb{Z}}} |n|^p |\langle \xi ,e_n \rangle|^2 = \infty\}\]
is a dense $G_\delta$ set in ${\mathcal{H}}$, where ${\mathbb{Q}}^+:=\{x\in {\mathbb{Q}}\mid x>0\}$.
\end{proof}

%%%%%%%%%%%%%%%%%%%%%%%%%%%%%%%%%%%%%%%%%%%%%%%%%%%%%%%%%%%%%%%%%%%%%%%%--Acknowledgments--and--bibliography--%%%%%%%%%%%%%%%%%%%%%%%%%%%%%%%%%%%%%%%%%%%%%%%%%%%%%%%%%%%%%%%%%%%%%%%%%%%%%%%%%%%%%%%%%%%%%%%%%%%%%%%%%%%%%%%%%%%%%%%%%%%%%%%%%%%%%%%%%%%%%%%%%%%%%%%%%%%%%%%%%%%%%%%%%%%%%%%%%%%%%%%%%%%%%%%%%%%%%%%%%%%%%%%%%%%%%%%%%%%%%%%%%%%%%%%%%%%%%%%%%%%%%%%%%%%%%%%%%%%%%%%%%%%%%%%%%%%%%%%%%%%%%%%%%%%%%%%%%%%%%%%%%%%%%%%%%%%%%%%%%%%%%%%%%%%%%%%%%%%%%%%%%%%%%%%%%%%%%%%%%%%%%%%%%%%%%%%%%%%

\begin{center} \Large{Acknowledgments} 
\end{center}

M.A.\ was supported by CAPES (a Brazilian government agency). S.L.C.\ thanks the to partial support by FAPEMIG (a Brazilian government agency; Universal Project 001/17/CEX-APQ-00352-17).

\end{document}